\documentclass[superscriptaddress, reprint, amsmath, amssymb, aps, pra, floatfix]{revtex4-2}

\usepackage[utf8]{inputenc}
\usepackage{graphicx}
\usepackage{dcolumn}
\usepackage{bm}
\usepackage[normalem]{ulem} 
\usepackage{mathrsfs}
\usepackage{physics}
\usepackage{xcolor}
\usepackage{amsmath, amssymb, mathtools, amsthm}
\usepackage[colorlinks=true]{hyperref}
\usepackage{comment} 
\usepackage{orcidlink}

\newtheorem{theorem}{Theorem}
\newtheorem{lemma}{Lemma}
\theoremstyle{remark}

\newcommand{\RR}{\mathcal{R}}
\newcommand{\TT}{\mathcal{T}}
\newcommand{\id}{\mathbb{I}}
\newcommand{\e}{\mathrm{e}}
\newcommand{\asympL}{x\to-\infty}
\newcommand{\asympR}{x\to+\infty}

\begin{document}

\title{Klein tunneling through an asymmetric barrier: Symmetric transmission and directional pair creation}

\author{Andre G. Campos\orcidlink{0000-0003-2923-4647}}
\affiliation{Max Planck Institute for Nuclear Physics, Heidelberg 69117, Germany}
\email{agontijo@mpi-hd.mpg.de}

\author{Denys I. Bondar~\orcidlink{0000-0002-3626-4804}}
\email{dbondar@tulane.edu}
\affiliation{Department of Physics and Engineering Physics, Tulane University, New Orleans, Louisiana 70118, United States}

\begin{abstract}
We prove that the transmission probability for the Klein tunneling through a spatially asymmetric barrier is the same for left and right incidence whenever each asymptotic lead carries a single propagating channel per direction. Time-dependent Wigner-function simulations confirm this and locate the missing directionality in the barrier's interior, where a sharp edge generates several times more under-barrier negative-energy population than a smooth one. Directional control in the Klein regime therefore resides in pair production rather than in the transmitted current.
\end{abstract}

\date{\today}

\maketitle

\section{Introduction}

A quantum particle that strikes a one-dimensional potential from the left and one that strikes it from the right produce scattering states that look nothing alike, yet the two share a striking property: the reflection and transmission probabilities are identical, even when the potential possesses no spatial symmetry whatsoever. This was noted by Landau and Lifshitz~\cite{LandauLifshitz}, and later in Ref.~\cite{CohenTannoudji}. Shegelski and Sample gave a general and remarkably elegant proof, valid for any finite potential tending to constant values as $x\to\pm\infty$~\cite{ShegelskiSample2020}. This symmetry is not, however, a universal property of a single particle in one dimension: its violation for a potential with an asymptotically linear ramp, $V(x)\propto x$ as $x\to\pm\infty$, was shown in Sec.~V of Ref.~\cite{schach_tunneling_2022}.

Shortly after the Landau-Lifshitz observation, Amirkhanov and Zakhariev~\cite{Amirkhanov1966} showed that the tunneling symmetry can be broken for a  composite particle. They studied a two-particle system resolved into a center-of-mass (C.M.) coordinate and an inter-particle degree of freedom (for illustation see Fig.~2 of Ref.~\cite{Lindberg2023}). Consider a triangular barrier and prepare the inter-particle degree of freedom in its ground state. If the packet meets the gently sloping side, the barrier builds up adiabatically, the internal state is not excited, and the dynamics reduces to the ordinary one-dimensional problem. If instead the packet meets the vertical face, the sudden shakeup drives a transition, and since the internal state was already the lowest one, the only available transition is upward. The excitation energy must be taken out of the C.M. kinetic energy; the C.M. therefore plunges deeper beneath the barrier and transmission is suppressed relative to the opposite orientation. This simple argument, long overlooked, has proved unusually productive: it underlies the demonstration that barrier symmetry can be engineered to enhance or suppress tunneling, to the point that tunneling through a barrier becomes more probable than passing over it~\cite{BondarLiuIvanov2010}. It also has inspired the proposal of asymmetric tunneling of Bose--Einstein condensate~\cite{Lindberg2023}, unique many-body resonant transport phenomena for fermions~\cite{Bilokon2025, SpinRing}, and a black-hole analogue for bosons~\cite{BlackHole}.

Both threads above are nonrelativistic. A distinctive feature of the relativistic Dirac equation, entirely absent from the Schr\"odinger case, is the presence of negative-energy solutions, interpreted as antiparticles. An external field couples the two branches of the spectrum. For a  barrier of height below $mc^2$, this population this coupling is negligible and the dynamics qualitatively reduces to the Schr\"odinger problem, for which left--right symmetry of tunneling is well established. This changes if we work with a high barrier. Klein~\cite{Klein1929} found that a Dirac electron striking a step of height $V_0 > E + mc^2$ is not exponentially attenuated but crosses with unit probability. The barrier pushes the positive-energy continuum into the negative-energy continuum, so the particle traverses the barrier as a negative-energy state and re-emerges as a particle~\cite{Cabrera2016}. The Klein tunneling is an interband transition mediated by the spinor degree of freedom, and has been observed in graphene~\cite{KatsnelsonGeim2006,CastroNeto2009,Allain2011}. Since an internal degree of freedom is coupled to the translational motion,  it is natural to investigate whether the Klein tunneling remains left--right symmetric for an asymmetric barrier.

This question sits at an apparent impasse. On one side stands a scattering-matrix argument, standard in mesoscopic transport~\cite{Buttiker1986,Datta1995,Beenakker1997}. The Dirac Hamiltonian is Hermitian and the Dirac current is conserved, so if each asymptotic lead supports exactly one propagating channel in each direction, the scattering data assemble into a $2\times2$ matrix $S$ that is unitary in the flux-normalized basis. Normalizing the columns of $S$ gives $|r_L|^2+|t_L|^2=1$ and $|r_R|^2+|t_R|^2=1$, while normalizing the rows gives $|r_L|^2+|t_R|^2=1$; subtracting these yields $|t_L|=|t_R|$ and $|r_L|=|r_R|$.  On this reading, Klein tunneling ought to be strictly symmetric, however asymmetric the barrier. On the other side stands the argument of Amirkhanov and Zakhariev~\cite{Amirkhanov1966}. A Dirac particle is not a structureless point in one dimension: alongside the coordinate $x$ it carries a spinor degree of freedom, and a spatially varying field couples that internal space to the translational motion in the same structural way that the barrier couples the inter-particle coordinate of a composite system. The analogy is not exact, but the ingredient that causes the asymmetry, namely a sudden gradient acting on an internal degree of freedom, is present. The two arguments cannot both be right as stated.

This paper resolves the dichotomy, and the resolution turns on \emph{where} the internal degree of freedom is dynamically available. In Sec.~\ref{SecProof} we adapt the Shegelski--Sample proof~\cite{ShegelskiSample2020} to the Dirac equation and confirm the scattering-matrix conclusion: the transmission probability is the same from either side. The Amirkhanov--Zakhariev mechanism cannot overturn this because, asymptotically, there is no second state for it to act on. Where the potential has settled to $V_a$, the dispersion $(E-V_a)^2=m^2c^4+\hbar^2c^2k^2$ admits only two propagating modes at a given energy, one of each current sign, with spinors fixed by $(E,k)$. The sign of $E-V_a$ decides which continuum they belong to, and the two continua are never open at once in the same lead. There is thus no asymptotic level to be promoted into. What the Klein regime does provide is an \emph{interior} window: under the barrier region supports  negative-energy modes, with no asymptotic counterpart, that can be excited only locally. The asymmetry that Klein tunneling cannot express in its transmission probability is expelled into this interior sector, where it survives as a directionally dependent population of negative-energy states. Our time-dependent Wigner-function simulations in Sec.~\ref{SecWignerSimulations} confirm this: Transmission is insensitive to the side of incidence, while a steep barrier edge generates several times more under-barrier negative-energy population than a smooth one.

\section{Time-independent analysis}\label{SecProof}

Consider the stationary Dirac equation
\begin{equation}
    H_D\Psi(x)=E\Psi(x),
    \qquad
    H_D=-i\hbar c\,\alpha\frac{\dd}{\dd x}+M(x),
    \label{eq:dirac-general}
\end{equation}
where
\begin{equation}
    \alpha^\dagger=\alpha,
    \qquad
    \alpha^2=\id,
    \qquad
    M(x)^\dagger=M(x).
\end{equation}
Assume that
\begin{align}
    M(x)\longrightarrow M_L \quad \text{as} \quad x\to-\infty, 
    \notag\\
    M(x)\longrightarrow M_R \quad \text{as} \quad x\to+\infty,
    \label{eq:asymptotic-M}
\end{align}
with constant Hermitian matrices $M_L$ and $M_R$.

The usual two-component Dirac equation with a real electrostatic potential is
included as the special case
\begin{equation}
    H_D=-i\hbar c\,\sigma_x\frac{\dd}{\dd x}
        +mc^2\sigma_z+V(x).
    \label{eq:electrostatic-dirac}
\end{equation}

At the chosen energy $E$, suppose that each asymptotic region supports
exactly one propagating mode carrying current toward $+x$ and exactly one
propagating mode carrying current toward $-x$.  Threshold energies, at which a
propagating mode has zero current, are excluded.  Closed-channel terms that
decay at spatial infinity may be present, but they do not affect the
asymptotic fluxes below.

For two stationary solutions $\Phi$ and $\Psi$ at the same energy $E$, define
the mixed current
\begin{equation}
    J(\Phi,\Psi)=c\,\Phi^\dagger\alpha\Psi.
    \label{eq:mixed-current}
\end{equation}
The ordinary probability current is the quadratic form
\begin{equation}
    j[\Psi]=J(\Psi,\Psi)=c\,\Psi^\dagger\alpha\Psi.
    \label{eq:dirac-current}
\end{equation}

\begin{lemma}\label{lemma1}
For any two solutions $\Phi$ and $\Psi$ of Eq.~\eqref{eq:dirac-general} at the
same real energy, $J(\Phi,\Psi)$ is independent of $x$.
\end{lemma}

\begin{proof}
Equation~\eqref{eq:dirac-general} gives
\begin{equation}
    \frac{\dd \Psi}{\dd x}
    =\frac{i}{\hbar c}\,\alpha(E-M)\Psi,
    \quad
    \frac{\dd \Phi'^\dagger}{\dd x} 
    =-\frac{i}{\hbar c}\,\Phi^\dagger(E-M)\alpha.
\end{equation}
Therefore,
\begin{align}
    \frac{\dd}{\dd x}J(\Phi,\Psi)
    &=c\,\Phi'^\dagger\alpha\Psi
      +c\,\Phi^\dagger\alpha\Psi' \notag\\
    &=-\frac{i}{\hbar}\Phi^\dagger(E-M)\Psi
      +\frac{i}{\hbar}\Phi^\dagger(E-M)\Psi = 0. \notag
\end{align}
In particular, the current $j[\Psi]$ is constant.
\end{proof}

Let $\phi_{a,+}$ and $\phi_{a,-}$ denote the right-current and left-current
propagating modes in lead $a\in\{L,R\}$.  The labels $+$ and $-$ refer to the direction of current, not necessarily to
the sign of momentum or to the sign of the energy relative to an asymptotic
potential.
Normalize them to unit flux:
\begin{equation}
    j[\phi_{a,+}]=+1,
    \qquad
    j[\phi_{a,-}]=-1.
    \label{eq:unit-flux}
\end{equation}
Away from a threshold, the two modes can also be chosen to satisfy
\begin{equation}
    J(\phi_{a,+},\phi_{a,-})=0.
    \label{eq:flux-orthogonality}
\end{equation}
Indeed, if
\begin{equation}
    \phi_{a,\pm}(x)=u_{a,\pm}\e^{ik_{a,\pm}x},
\end{equation}
then
\begin{equation}
    J(\phi_{a,+},\phi_{a,-})
    =c\,u_{a,+}^\dagger\alpha u_{a,-}
      \e^{i(k_{a,-}-k_{a,+})x}.
\end{equation}
According to Lemma~\ref{lemma1}, this mixed current must be independent of $x$, whereas
$k_{a,+}\neq k_{a,-}$ for two distinct nonthreshold propagating modes; hence, $c = 0$ proving
Eq.~\eqref{eq:flux-orthogonality}.

It follows that an asymptotic superposition has current
\begin{equation}
    j[A\phi_{a,+}+B\phi_{a,-}]=|A|^2-|B|^2.
    \label{eq:asymptotic-current}
\end{equation}

Choose unit incoming amplitudes in the unit-flux basis.  The state incident
from the left has the asymptotic form
\begin{equation}
    \Psi_L(x)\sim
    \begin{cases}
        \phi_{L,+}+r_L\phi_{L,-}, & \asympL,\\[2mm]
        t_L\phi_{R,+},            & \asympR,
    \end{cases}
    \label{eq:left-incidence}
\end{equation}
and the state incident from the right has the form
\begin{equation}
    \Psi_R(x)\sim
    \begin{cases}
        t_R\phi_{L,-},                         & \asympL,\\[2mm]
        \phi_{R,-}+r_R\phi_{R,+},              & \asympR.
    \end{cases}
    \label{eq:right-incidence}
\end{equation}
Here $r_L,t_L$ are the reflection and transmission amplitudes for left
incidence, and $r_R,t_R$ are those for right incidence.

Because the modes are flux normalized, the correspining reflection and transmission probabilities are defined
\begin{equation}
    \RR_L=|r_L|^2,
    \,
    \TT_L=|t_L|^2,
    \,
    \RR_R=|r_R|^2,
    \,
    \TT_R=|t_R|^2.
    \label{eq:probabilities}
\end{equation}
Current conservation for Eqs.~\eqref{eq:left-incidence} and
\eqref{eq:right-incidence} gives
\begin{equation}
    1-|r_L|^2=|t_L|^2,
    \qquad
    1-|r_R|^2=|t_R|^2.
    \label{eq:individual-unitarity}
\end{equation}

\begin{theorem}[Equal reflection and transmission probabilities for the Dirac equation]\label{MainTheorem}
Under the assumptions of Eqs.~\eqref{eq:left-incidence} and \eqref{eq:right-incidence},
\begin{equation}\label{eq:maintheorem}
    \RR_L(E)=\RR_R(E),
    \qquad
    \TT_L(E)=\TT_R(E).
\end{equation}
No parity symmetry of $M(x)$ in Eq.~\eqref{eq:dirac-general} is required.
\end{theorem}

\begin{proof}
Linearity of the stationary Dirac equation implies that any constant linear
combination of $\Psi_L$ and $\Psi_R$ is again a solution at energy $E$.
Define
\begin{equation}
    \Delta:=t_Lt_R-r_Lr_R.
    \label{eq:Delta}
\end{equation}

First form the combination
\begin{equation}
    \Phi_A:=t_L\Psi_R-r_R\Psi_L.
    \label{eq:PhiA}
\end{equation}
Using Eqs.~\eqref{eq:left-incidence} and \eqref{eq:right-incidence}, its
asymptotic form on the right is
\begin{equation}
    \Phi_A\sim t_L\phi_{R,-},
    \qquad \asympR,
    \label{eq:PhiA-right}
\end{equation}
because the two terms proportional to $\phi_{R,+}$ cancel.  On the left,
\begin{equation}
    \Phi_A\sim-r_R\phi_{L,+}+\Delta\phi_{L,-},
    \qquad \asympL.
    \label{eq:PhiA-left}
\end{equation}
Equation~\eqref{eq:asymptotic-current} therefore gives
\begin{equation}
    j_A(+\infty)=-|t_L|^2,
    \qquad
    j_A(-\infty)=|r_R|^2-|\Delta|^2.
\end{equation}
Conservation of the current of $\Phi_A$ yields
\begin{equation}
    |\Delta|^2=|r_R|^2+|t_L|^2.
    \label{eq:identity-A}
\end{equation}

Next form
\begin{equation}
    \Phi_B:=t_R\Psi_L-r_L\Psi_R.
    \label{eq:PhiB}
\end{equation}
The left-going terms cancel on the left, so
\begin{equation}
    \Phi_B\sim t_R\phi_{L,+},
    \qquad \asympL.
    \label{eq:PhiB-left}
\end{equation}
On the right,
\begin{equation}
    \Phi_B\sim\Delta\phi_{R,+}-r_L\phi_{R,-},
    \qquad \asympR.
    \label{eq:PhiB-right}
\end{equation}
Hence
\begin{equation}
    j_B(-\infty)=|t_R|^2,
    \qquad
    j_B(+\infty)=|\Delta|^2-|r_L|^2.
\end{equation}
Conservation of the current of $\Phi_B$ gives
\begin{equation}
    |\Delta|^2=|r_L|^2+|t_R|^2.
    \label{eq:identity-B}
\end{equation}

Equating Eqs.~\eqref{eq:identity-A} and \eqref{eq:identity-B},
\begin{equation}
    |r_R|^2+|t_L|^2=|r_L|^2+|t_R|^2.
    \label{eq:combined-identity}
\end{equation}
Now use Eq.~\eqref{eq:individual-unitarity} to replace
$|t_L|^2$ by $1-|r_L|^2$ and $|t_R|^2$ by $1-|r_R|^2$.  Then
\begin{equation}
    |r_R|^2+1-|r_L|^2
    =|r_L|^2+1-|r_R|^2,
\end{equation}
which implies
\begin{equation}
    |r_L|^2=|r_R|^2.
\end{equation}
A final use of Eq.~\eqref{eq:individual-unitarity} gives
\begin{equation}
    |t_L|^2=|t_R|^2.
\end{equation}
Together with Eq.~\eqref{eq:probabilities}, these are precisely~Eq.~\eqref{eq:maintheorem}.
\end{proof}

Note that time-reversal symmetry was not used.  The result follows from linearity,
Hermiticity, current conservation, and the fact that the asymptotic scattering
problem has only two open ports with one channel per port. When the two asymptotic states lie in different Dirac continua, incoming and outgoing modes
must be identified by the sign of their current, not by the sign of their
momentum.  The stationary one-particle flux proof remains valid whenever the
stated open-channel assumptions hold.  

\section{Time-dependent analysis}\label{SecWignerSimulations}

Theorem~\ref{MainTheorem} constrains only what the asymptotic leads can see: the flux carried by stationary states at a fixed energy, resolved into the two open channels of each lead. It is silent about the interior of the scattering region, which in the Klein regime supports negative-energy modes that have no asymptotic counterpart and can therefore be excited only locally. If the argument of Sec.~\ref{SecProof} is correct, the barrier asymmetry must leave the transmission untouched and reappear somewhere else; the natural candidate is that interior sector. We now test both halves of this expectation by following the dynamics in real time.

The natural language for what follows is phase space. The Wigner function is the quasiprobability distribution obtained by Fourier transforming the density matrix in the relative coordinate \cite{Wigner1932,Hillery1984,schleich2015quantum}: It is real and its marginals reproduce the position and momentum densities, but it is not positive definite, and negativity signals quantum interference. For the Dirac equation the construction yields a matrix-valued distribution \cite{ElzeGyulassyVasak1986,VasakGyulassyElze1987,BialynickiBirula1991}, whose scalar component $W(x,p)=\mathrm{Tr}[\mathcal{W}(x,p)\gamma^{0}]/4$ retains both marginals and is the object plotted below \cite{Cabrera2016}. The relativistic case does carry one caveat:  positivity alone no longer certifies the absence of interference \cite{Campos2014}.

We solve the time-dependent Dirac equation in phase space, propagating $W(x,p)$ with the split-operator propagator~\cite{Cabrera2016}. Throughout this section we work in natural units, $\hbar=c=m=1$, so that the rest energy is $mc^2=1$.

The scatterer is the following spatially asymmetric barrier
\begin{equation}
 V(x)   \frac{9}{16} \left( 20 e^{-\frac{2(x+4)^2}{21}} + 7 e^{-\frac{2x^2}{21}} + 3 e^{-\frac{2(x-4)^2}{21}} \right),
\end{equation}
which is shwon in Fig.~\ref{fig_potential}. It rises abruptly on one side and decays over an extended, smooth tail on the other, so that the two leads are reached through edges of very different steepness, while the asymptotic potentials themselves coincide, $V(\asympL)=V(\asympR)=0$, as required by the assumptions of Sec.~\ref{SecProof}.

\begin{figure}[t]
    \centering
    \includegraphics[width=\columnwidth]{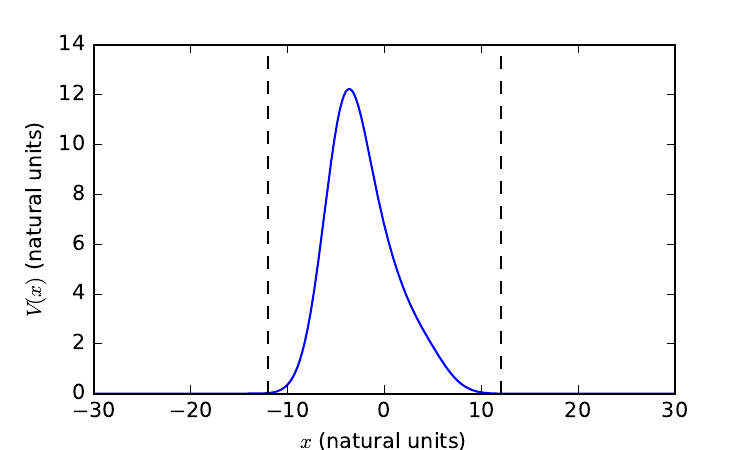}
    \caption{The spatially asymmetric barrier $V(x)$ used in the simulations, in natural units. One edge is sharp and the other decays over an extended smooth tail, while the asymptotic values on the two sides coincide, $V(\asympL)=V(\asympR)=0$. The barrier height exceeds the rest energy $mc^2=1$, so the dynamics lies in the Klein regime. The dashed lines at $x=\pm12$ delimit the interior region; $|x|>12$ defines the reflected and transmitted populations quoted in the text.}
    \label{fig_potential}
\end{figure}

The initial state is a Gaussian wave packet built on a positive-energy spinor, of the form used in Ref.~\cite{Cabrera2016}. Its central momentum and spatial width are chosen so that the mean energy lies well below the top of the barrier and the momentum spread is narrow on the scale of the barrier height. The dynamics is therefore in the tunneling regime, with a transmitted population of a few percent [Fig.~\ref{fig_time_observable}(a)]. The corresponding phase-space portrait, Fig.~\ref{fig_wigner}(a), is a single compact positive Gaussian displaced toward the barrier along the momentum axis.

The barrier of Fig.~\ref{fig_potential} is held fixed throughout; the two propagations differ only in where the initial packet is prepared and in which direction it is launched, so that it encounters the sharp edge first in one run and the smooth edge first in the other. In the first run the packet starts to the left of the barrier and moves toward $+x$: we then call the probability of finding the particle at $x<-12$ the reflection probability and the probability of finding it at $x>12$ the tunneling probability (see Fig.~\ref{fig_time_observable}). In the second run the packet starts to the right of the barrier and moves toward $-x$, and the two roles are interchanged, the population at $x>12$ being the reflected one and that at $x<-12$ the transmitted one. The two runs are the wave-packet analogues of the left- and right-incident stationary states $\Psi_L$ and $\Psi_R$ of Eqs.~\eqref{eq:left-incidence} and \eqref{eq:right-incidence}. Alongside these populations we monitor the total weight of the negative-energy states, obtained by projecting the propagated state onto the negative-energy subspace at each instant.

\begin{figure}[t]
    \centering
    \includegraphics[width=\columnwidth]{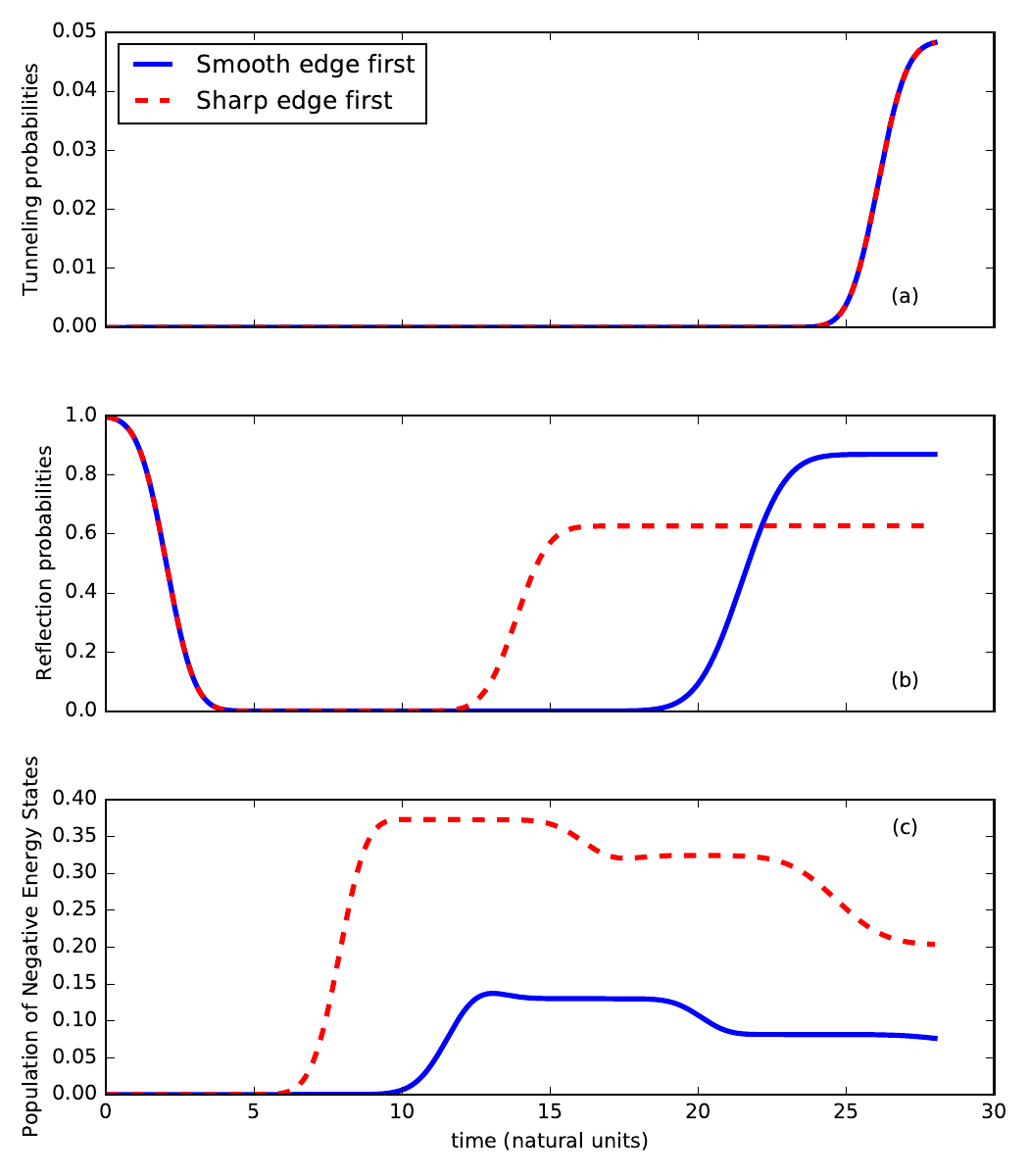}
    \caption{Time-resolved observables for the two propagations, one in which the wave packet meets the sharp edge of the barrier of Fig.~\ref{fig_potential} first and one in which it meets the smooth edge first. (a) Tunneling probability, i.e., the population beyond the barrier. The two curves coincide, confirming Theorem~\ref{MainTheorem} (b) Reflection probability, i.e., the population in front of the barrier. The two curves separate. (c) Population of the negative-energy states. The sharp edge generates about four times more antiparticles than the smooth edge; this population remains under the barrier and accounts for the difference in panel (b).}
    \label{fig_time_observable}
\end{figure}

The results are displayed in Fig.~\ref{fig_time_observable}. Panel (a) shows the tunneling probability, defined as the probability of finding the particle after the barrier. In agreement with the theorem of Sec.~\ref{SecProof}, the two curves are indistinguishable throughout the evolution: the transmitted population is insensitive to the side of incidence, despite the pronounced asymmetry of the scatterer.

The interior behaves in the opposite way. Figure~\ref{fig_time_observable}(c) shows the population of the negative-energy states: about four times as many antiparticles are created when the packet strikes the sharp edge as when it strikes the smooth one, which is what the known physics of Klein tunneling leads one to expect, pair creation being controlled by the steepness of the potential barrier. This population does not leave the scattering region on the time scale of the simulation, and its imprint is visible in Fig.~\ref{fig_time_observable}(b): the probability of finding the particle in front of the barrier differs between the two runs by an amount that tracks the excess negative-energy weight generated at the sharp edge. The flux that is stored in the interior sector is flux that has not yet come back out.

\begin{figure}[t]
    \centering
    \includegraphics[width=\columnwidth]{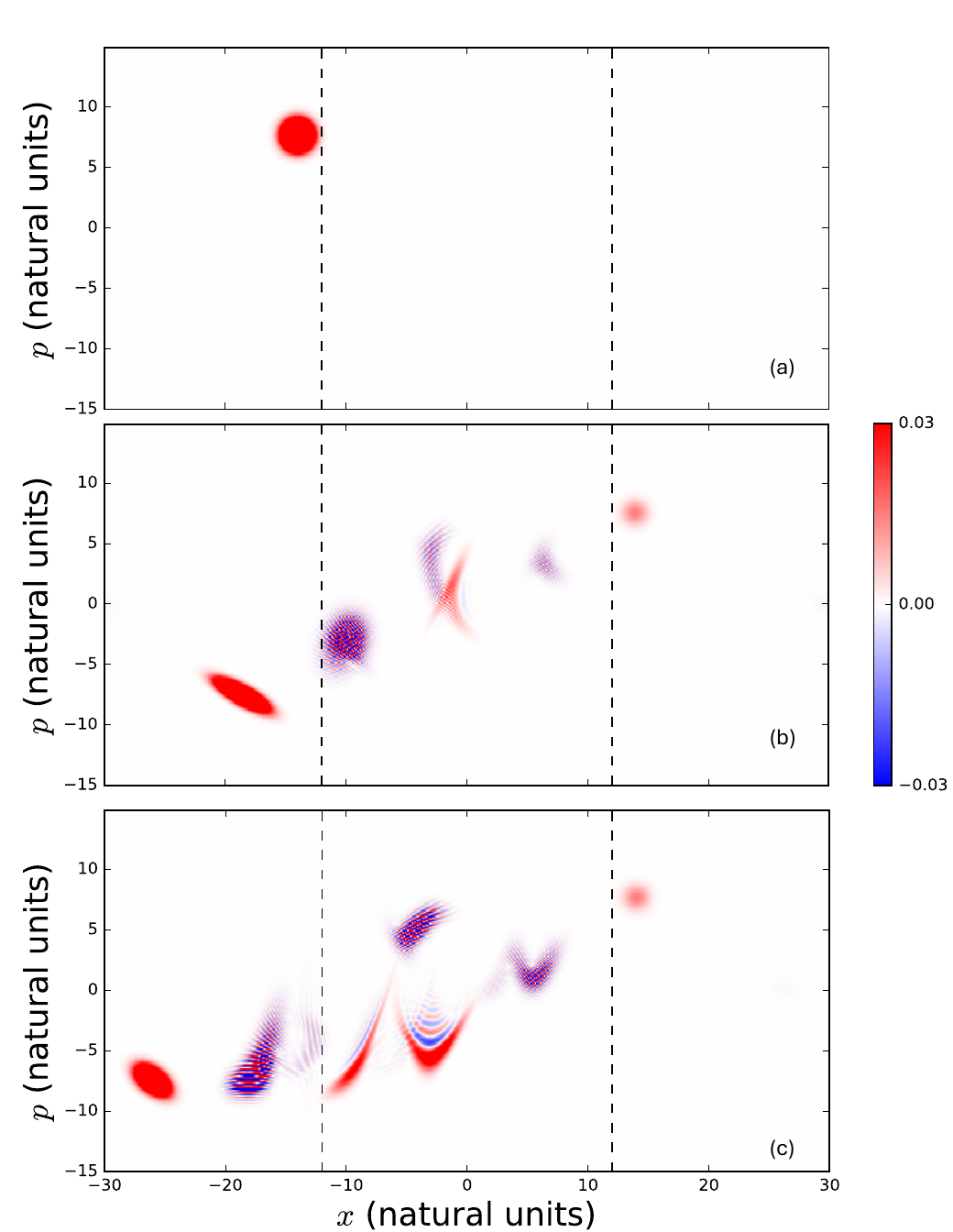}
    \caption{Phase-space portrait of the two propagations, displayed as the Wigner function $W(x,p)$ in natural units. (a) Initial state, in the geometry in which it faces the smooth edge. (b) Final state of that run. (c) Final state of the run facing the sharp edge, plotted after the phase-space inversion $W(x,p)\to W(-x,-p)$, which maps the mirrored geometry onto the axes of panel (b) and allows the two final states to be compared directly. The dashed lines at $x=\pm12$ delimit the interior region entering Fig.~\ref{fig_time_observable}. The sharp edge generates a rich interference pattern, the phase-space signature of \emph{zitterbewegung}, concentrated underneath the barrier.}
    \label{fig_wigner}
\end{figure}

The mechanism becomes transparent in phase space, Fig.~\ref{fig_wigner}. Panels (b) and (c) show the two final states on common axes, the second run being displayed after the inversion $W(x,p)\to W(-x,-p)$ mapping the right-incident simulation onto the axes of the left-incident simulation. After the barrier, $x > 12$, the transmitted distributions not only look identical, but carry the same probability, in accordance with Theorem~\ref{MainTheorem}. Before and under it they do not.  The initial Gaussian state that struck the smooth edge of the barrier remains closer to the barrier, whereas the Gaussian state that struck the sharp edge develops a rapidly oscillating pattern with pronounced negative regions. These fringes are the phase-space signature of \emph{zitterbewegung}, i.e., of the interference between positive- and negative-energy components that a steep barrier generates, and they sit in the region $|x|<12$, precisely where the missing flux of Fig.~\ref{fig_time_observable}(b) resides.

The simulations therefore reproduce both halves of the picture anticipated above. The transmission obeys the theorem exactly, as it must. The asymmetry of the barrier is not absent from the dynamics; it is displaced into the interior, where it survives as a directionally dependent population of negative-energy states. That the finite-time reflected populations differ is not a violation of the theorem but a consequence of this displacement: a wave packet is a superposition of stationary states, and the weight that remains bound under the barrier has not yet been resolved into the asymptotic outgoing channels to which the theorem applies.

\section{Conclusion}

We have shown that transmission in one spatial dimension is a poor probe of barrier asymmetry for a Dirac particle, as long as the leads remain single-channel; what responds is the finite-time reflected population, which differs between the two orientations by the excess antiparticle created. Directional control must therefore be sought in the pair-production channel rather than in the transmitted current.

Two possible directions follow naturally. The first is to relax the boundedness of the potential and admit an asymptotically unbounded ramp, $V(x)\propto x$ as $x\to\pm\infty$, which already breaks the left--right symmetry in the nonrelativistic case~\cite{schach_tunneling_2022}. Such a ramp may violate the single-channel assumptions on which Theorem~\ref{MainTheorem} rests, since the asymptotic dispersion no longer settles to a fixed pair of propagating modes, and the interior asymmetry identified here would then have an open channel to escape into. The second is to explore the connection to the solid-state phenomenon that goes under the name of ``asymmetric Klein tunneling''~\cite{Li2017,Iurov2020,Zhang2021}. There the asymmetry resides in an anisotropic dispersion relation  and the barrier is symmetric, whereas here the dispersion is isotropic and the barrier is not. The two mechanisms are independent, and a platform carrying both may exhibit a directional dependence that neither produces alone.

\acknowledgments

D.I.B. is grateful to Prof. Christoph Keitel of Max Planck Institute for Nuclear Physics in Heidelberg for a brief but very productive visit without which this work would not have been possible. 

D.I.B. is supported by Army Research Office (ARO) (grant W911NF-23-1-0288; program manager Dr.~James Joseph).  The views and conclusions contained in this document are those of the authors and should not be interpreted as representing the official policies, either expressed or implied, of ARO, or the U.S. Government. The U.S. Government is authorized to reproduce and distribute reprints for Government purposes notwithstanding any copyright notation herein.

\bibliography{refs}

\end{document}